\documentclass[11pt]{article}

\usepackage{fullpage, amsmath, amssymb, amsthm, bbm, color, dsfont}
\usepackage{graphicx,caption,subcaption,placeins}
\usepackage{setspace}

\makeatletter
\def\verbatim@font{\linespread{1}\normalfont\ttfamily}
\makeatother

\usepackage{natbib}
\bibpunct{(}{)}{;}{a}{,}{,}
\usepackage{enumerate}
\newtheorem{assumption}{Assumption}
\newtheorem{definition}{Definition}
\newtheorem{theorem}{Theorem}
\newtheorem{proposition}{Proposition}
\newtheorem{theorem*}{Theorem}
\newtheorem{proposition*}{Proposition}
\newtheorem{corollary*}{Corollary}

\newtheorem{lemma}{Lemma}

\newtheorem{remark}{Remark}

\newcommand{\pr}{\textup{pr}}

\usepackage[bookmarks=true,bookmarksnumbered=true, hypertexnames=false,breaklinks=true]{hyperref}

\usepackage{xr}
\title{Randomization inference for composite experiments with spillovers and peer effects}%
\author{Hui Xu, Guillaume Basse}
\begin{document}

\maketitle

\pagenumbering{gobble}

\begin{abstract}
 Group-formation experiments, in which experimental units are randomly
  assigned to groups, are a powerful tool for studying peer effects in the
  social sciences. Existing design and analysis approaches allow researchers
  to draw inference from such experiments without relying on parametric
  assumptions. In practice, however, group-formation experiments are often
  coupled with a second, external intervention, that is not accounted for
  by standard nonparametric approaches. This note shows how to construct
  Fisherian randomization tests and Neymanian asymptotic confidence
  intervals for such composite experiments, including in settings where
  the second intervention exhibits spillovers. We also propose an
  approach for designing optimal composite experiments.

\medskip 

\end{abstract}


\singlespacing
\clearpage
\pagenumbering{arabic}


\section{Introduction}
When studying social systems and organizations, quantitative researchers are
often interested in whether the behavior of an individual is affected by the
characteristics of other individuals in the system: this phenomenon is called a
peer effect. A common approach for studying peer effects is the so-called
group-formation experiment, whereby units are randomly split into groups. An
early example is a study conducted by \citet{sacerdote2001peer}, who leveraged
the random assignment of roommates at Dartmouth to assess whether the drinking
behavior of freshmen affected that of their roommates. In recent work,
\citet{li2019randomization} and \citet{basse2019randomization} developed a
framework for designing and analyzing these types of experiments in a randomization-based
framework; that is, without assuming a response model for the outcomes, and
relying on the random assignment as the sole basis for inference.

However, group-formation experiments are often coupled with an additional
intervention to form what we call a composite experiment: typically, units would
first be split into groups, then a treatment would be randomized to a subset of
the individuals in the experimental population. For instance, in their study of peer-effects in
the context of the spread of managerial best practices, \citet{cai2018interfirm}
randomized the managers of different-sized firms into groups, then provided a
random subset of the managers with special information. Similarly, in a study of
student learning, \citet{kimbrough2017peers} first randomized students into
groups of homogeneous or heterogenous ability, then allowed a random subset of
students to practice a task with other students in their group. Without the
second interventions, both studies would be simple group-formation experiments,
and could be analyzed with the framework of
\citet{basse2019randomization}. Similarly, if one conditions on the group
composition, then the second part of the composite experiment is just a
classical randomized experiment, and the effects of interest fit in the usual
causal inference framework.

This article shows how to study jointly the peer-effects and causal effects of
a composite experiment. Our key insight is that the effects of both
the group-formation and the additional intervention can be summarized into
an exposure, or effective treatment. In particular, this approach allows us
to accomodate the fact that the second intervention may exhibit spillover
effects. Building on the group theoretical framework of \citet{basse2019randomization},
we propose a class of designs that is amenable both to inference in the Neyman
model for exposure contrasts, and to conditional randomization tests that can
be implemented with simple permutations of the exposures. Within that class of
designs, we derive optimal designs by solving a simple integer programming
problem: in some simulation settings, we found that optimal designs increase
the power by 80 percentage points over valid but more naive designs.


\section{Setup and framework}
\label{section:setup}

\subsection{Composite experiments and potential outcomes framework}
\label{section:notation}

We consider $N$ units indexed by $i = 1, \ldots, N$, each with a fixed attribute
$A_i \in \mathcal{A}$, which are assigned to two successive interventions.  In
the first intervention, the group-formation intervention, the $N = m \times K$
units are randomly assigned to $K$ distinct groups of equal size $m$. Following
\citet{basse2019randomization}, we denote by $L_i \in \{1,\ldots,K\}$ the group
to which unit $i$ is assigned, and denote by $L = (L_1, \ldots, L_N)$ the group
assignment vector. For each group assignment vector $L$, define the neighbor
assignment vector $Z(L) = (Z_1(L),\ldots, Z_N(L))$, where
$Z_i(L) = \{j \in \mathbb{I} \setminus \{i\}: L_i = L_j\}, \forall i \in
\mathbb{I}.$ To simplify the notation, the dependence of $Z$ on $L$ will often
be omitted. In the second intervention, the treatment intervention, units are
randomly assigned to a treatment, with $W_i \in \mathcal{W}$ being a treatment
indicator for unit $i$ and $W = (W_1, \ldots, W_N)$ the treatment assignment
vector. We denote by $Y_i(Z, W)$ the potential outcome of unit $i$ which, a
priori, may depend on the entire group assignment vector $Z$ and treament
assignment vector $W$. Throughout, we will adopt the randomization-based
perspective, considering the potential outcomes as fixed quantities, the
randomness coming exclusively from $Z$ and $W$.  \citet{basse2019randomization}
studied the group-formation intervention, with no treatment intervention. In
contrast, the bulk of the literature on interference in causal inference focuses
on treatment interventions, without group-formation. Our setting combines both,
as summarized in the left panel of Figure~\ref{setup}, and allows us to address
a broader type of questions, as illustrated in the following examples.

\begin{figure}[!h]
\begin{minipage}[b]{0.55\textwidth}
\includegraphics[width=\textwidth]{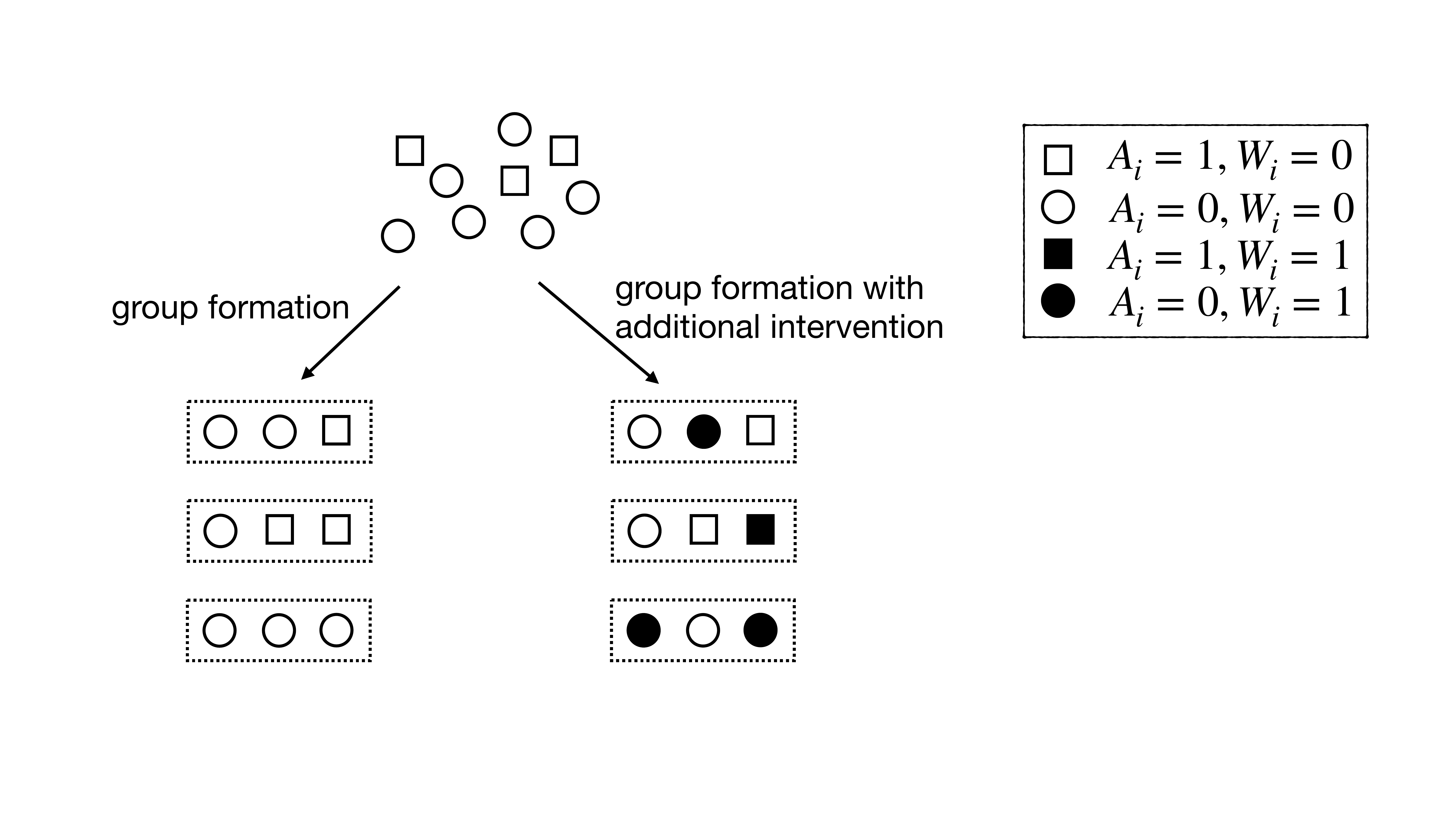}
\label{setup}
\end{minipage}
\begin{minipage}[b]{0.45\textwidth}
\includegraphics[width=\textwidth]{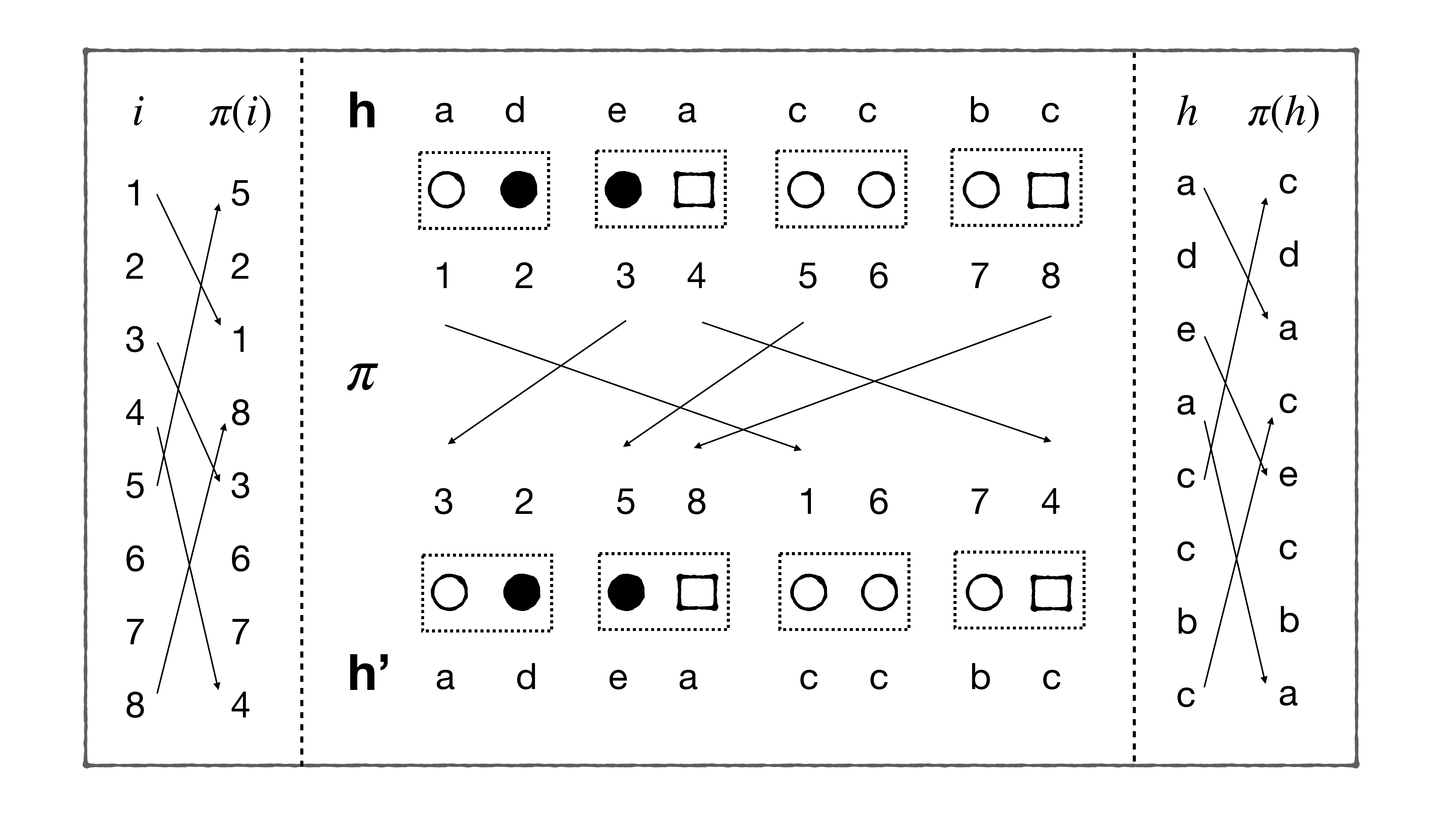}
\label{setup}
\end{minipage}
\caption{Left panel summarizes our composite experiment setting. Right panel illustrates our
  randomization procedure, as described in Definition~\ref{def:rand-procedure}.}
\label{setup}
\end{figure}

\medskip

\noindent \textbf{Example 1:}
In the managerial setting of \citet{cai2018interfirm} described
in the introduction, the attribute set $\mathcal{A}$ contains all the combinations
of size and sector for the firms, and $W_i$ is an indicator for whether the manager
of firm $i$ received special financial information.

\medskip

\noindent \textbf{Example 2:}
In the educational context of \cite{kimbrough2017peers} we mentioned earlier,
the attribute set $\mathcal{A}$ contains the different levels of student ability,
and $W_i$ is an indictor for whether student $i$ was allowed to practice a task
with another student.

\subsection{Exposure}

The potential outcomes notation $Y_i(Z, W)$ highlights the fact that the outcome of unit $i$
may depend on the group membership of all units, $Z$, as well as the treatment assigned to
all units, $W$. In practice, it is often reasonable to assume that the outcome of unit $i$
only depends on the treatments and attributes of the units in the same group as unit $i$;
that is, the outcome of unit $i$ depends on $Z$ and $W$ only through the function $h_i(Z, W)$
defined as:
\begin{equation}\label{eq:exp-mapping}
  h_i(Z, W) = (\Gamma_i, W_i)
\end{equation}
where $\Gamma_i = \{(A_j, W_j) : j \in Z_i\}$. In the pure group-formation intervention,
as well as in the pure treatment intervention settings, a collection of functions $h_i$
summarizing $Z$ or $W$ is called an exposure mapping: we will adopt this terminology
as well. The local dependence captured by the specification of \eqref{eq:exp-mapping}
generalizes the concept of partial interference \cite{} which, in our context, can
be formulated as follows:
%
\begin{assumption}\label{assumption}
  Let $\{h_i\}$ be as in \eqref{eq:exp-mapping}. For all $i = 1,\ldots, N$, the
  following holds:
  \begin{equation*}
    \forall (Z,W), (Z',W'),
    \qquad h_i(Z, W) = h_i(Z', W') \quad \Rightarrow \quad Y_i(Z,W) = Y_i(Z',W').
  \end{equation*}
  With a slight abuse of notation, we will write $Y_i(Z, W) = Y_i(\Gamma_i, W_i)$.
\end{assumption}
If we think of the pair $(Z, W)$ as the intervention, the exposure
$H_i = (\Gamma_i, W_i)$ can be thought of as the effective intervention, since it
captures the part of $(Z, W)$ that actually affects the outcome of unit $i$. When
both the attribute set $\mathcal{A}$ and treatment set $\mathcal{W}$ are binary, the
exposure of \eqref{eq:exp-mapping} simplifies to:
\begin{equation}\label{eq:binary-exposure-mapping}
  h_i(Z, W) = (\Gamma_i, W_i)
  = (\sum_{j\in Z_i} A_j, \sum_{j\in Z_i} Z_j, \sum_{j\in Z_i} A_jW_j, W_i)
\end{equation}
so the exposure $H_i$ can be summarized by a simple quadruple of values.

In practice, further restrictions of the exposure may be considered. For instance, one
may assume that the interaction term $\sum_{j\in Z_i} A_j W_j$ does not affect the
outcome, and can be removed from the exposure. While our results are derived for the
more general exposure, they can be shown to hold for this simplified exposure as well. 

\subsection{Causal estimands and null hypotheses}

We will consider two types of inferential targets, requiring two different
approaches to inference. First we will consider causal estimands defined as
average contrasts between different exposures. Specifically, for $\{h_i\}_{i=1}^n$
defined as in \eqref{eq:exp-mapping}, let $\mathcal{H}$ be the set of all
values that the exposures can take; since $H_i = (\Gamma_i, W_i)$, each element
$k \in \mathcal{H}$ will be of the form $k = (\gamma, w)$. We consider the
average exposure contrast between $k,k'\in\mathcal{H}$, defined as
$\tau(k, k') = N^{-1}\sum_{i=1}^N\{Y_i(k) - Y_i(k')\}$, as well as the
attribute-specific counterpart defined as
$\tau_{[a]}(k, k') = N_{[a]}^{-1} \sum_{i: A_i=a} \{Y_i(k) - Y_i(k')\}$,
where $N_{[a]}$ is the number of units with attribute $A_i = a$.
Two special cases of these estimands deserve a brief mention. If $k = (\gamma, w)$
and $k' = (\gamma', w')$ are such that $w = w'$, then the estimand focuses on the
effect of peer's attributes and treatments. If $k$ and $k'$ are such that
$\gamma' = \gamma$, then the estimand focuses on the effect of each unit's treatment,
for fixed levels of peer attributes and peer treatments. Second, we will consider
two types of null hypotheses. The global null hypothesis
\begin{equation*}
H_0: Y_i(\gamma_1,w_1) = Y_i(\gamma_2,w_2), \forall (\gamma_1,w_1),(\gamma_2,w_2) \in \mathcal{H}, \forall i \in \mathbb{I}
\end{equation*}
asserts that the combined intervention has no effect whatsoever on any unit. Of more
practical interest are pairwise null hypotheses of the form
\begin{equation*}
H_0^{(\gamma_1,w_1),(\gamma_2,w_2)}: Y_i(\gamma_1,w_1) = Y_i(\gamma_2,w_2), \forall i \in \mathbb{I}.
\end{equation*}
The global null hypothesis can be easily tested with a standard Fisher Randomization
Test so we discuss it only in the Supplementary Material. We will focus instead on
pairwise null hypotheses, which are more difficult to test since they
are not sharp; that is, under the pairwise null, the observed outcomes do not determine
all the potential outcomes.

\subsection{Assignment mechanism and challenges}

In section~\ref{section:notation}, we stated that both the group assignment $L$ and the treatment $W$ were
assigned at random, but so far we have not discussed their distribution $\pr(L, W)$. In a
randomization-based framework, this distribution is the sole basis for inference, and must
be specified with care.

Building on an insight from \citet{basse2019randomization}, notice that if we assume that the
outcome of unit $i$ depends on $Z$ and $W$ only through the exposure
$H_i = h_i(Z(L), W)$, the problem reduces to a multi-arm trial on the exposure scale.
In particular, instead of $\pr(L, W)$, one should focus on $\pr(H)$, the distribution of
the exposure induced by $\pr(L, W)$.
If the distribution of $\pr(H)$ is simple, estimating exposure contrasts and testing
pairwise null hypotheses is straightforward. Unfortunately, the experimenter
can manipulate $\pr(H)$ only indirectly, via $\pr(L, W)$. The key objective of this
paper is to construct a class of designs $\pr(L, W)$ that induce simple exposure
distributions $\pr(H)$; specifically, we focus on designs for which the exposure has
a Stratified Completely Randomized Design.

\begin{definition}
\label{SCRD}
Without loss of generality, denote by $\mathcal{H}$ the set of possible exposures and $A$ an N-vector. Let $\mathbf{n}_A = (n_{a,h})_{a \in \mathcal{A}, h \in \mathcal{H}}$, such that $\sum_{h\in \mathcal{H}} n_{a,h} = N_{[a]}$, denote a vector of non-negative integers corresponding to number of units with each possible attribute and exposure combination. We say that a distribution of exposures $\pr(H)$ is a stratified completely randomized design denoted by $\operatorname{SCRD}(\mathbf{n}_{A})$ if the following two conditions are satisfied. 
\begin{enumerate}
\item After stratifying based on $A$, the exposure $H = (H_1,\ldots, H_N)$ is completely randomized. That is, (1) $\mathbb{P}(H_i = h) = \mathbb{P}(H_j = h)$ for all $h \in \mathcal{H}$ and $i,j \in \mathbb{I}$ such that $A_i = A_j$; (2) the number of units with exposure $h \in \mathcal{H}$ and stratum $a \in \mathcal{A}$ is $n_{a,h}$.
\item The exposure assignments across strata are independent. That is $\mathbb{P}(H_i = h_i | H_j = h_j) = \mathbb{P}(H_i = h_i)$ for all $h_i, h_j \in \mathcal{H}$ and $i,j \in \mathbb{I}$ such that $A_i \neq A_j$. 
\end{enumerate}
\end{definition}

This design is simple for two reasons. First, it is easy to sample from: this makes it
possible to perform suitably adapted Fisher Randomization Tests, a task that would
otherwise be computationally intractable \citep{basse2019randomization}. Second, it makes it
possible to obtain inferential results for standard estimators such as the difference in means.

\section{Randomization procedure and main theorem}
\label{randomization}

Our main result builds on the theory developed by \cite{basse2019randomization}, and can be summarized
in one sentence: if the design $\pr(L, W)$ has certain symmetry properties, so
will the exposure distribution $\pr(H)$. The right notion of symmetry can be
formulated using elementary concepts from algebraic group theory. 

Recall that a permutation of $\mathbb{I} = \{1,\ldots,N\}$ can be represented as
a one-to-one mapping from $\mathbb{I}$ to $\mathbb{I}$. The symmetric group
$S$ is the set of all permutations of $\mathbb{I}$. Let $C_i = (L_i, W_i)$ and
$C = (C_1, \ldots, C_N) \in \mathbb{C}$. If $\pi \in S$, we denote by
$\pi \cdot C = (C_{\pi^{-1}(i)})_{i=1}^N = (L_{\pi^{-1}(i)}, W_{\pi^{-1}(i)})_{i=1}^N$
the operation of permuting the elements of $C$. This mathematical operation called
a group action is defined more formally in the Supplement. Finally, if $C \in \mathbb{C}$,
and $\Pi \subseteq S$ is a subgroup of $S$, we define the stabilizer group of
$C$ in $\Pi$ as $\Pi_C = \{\pi \in \Pi \,:\, \pi \cdot C = C\}$. We can now
introduce our proposed procedure:

\begin{definition}\label{def:rand-procedure}
\label{procedure}
Given an observed attribute vector $A = (A_1,\ldots, A_N)$, consider the following randomization procedure. 
\begin{enumerate}
\item Initialize $C_0 = (W_0$, $L_0) \in \mathbb{C}$.
\item Permute $C = \pi \cdot C_0$, where $\pi \sim \operatorname{Unif}(S_A)$ 
\end{enumerate}
\end{definition}
Given a choice of $C_0 = (W_0, L_0)$ this procedure yields a design $\pr(L, W)$
with two important properties. First, it is easy to sample from: drawing random
permutations from $S_A$ and applying them to a vector $C_0$ can be done in just
three lines of efficient $R$ code, without requiring additional packages
\citep{basse2019randomization}. Second, it induces a simple exposure distribution, as
formalized by Theorem~\ref{main} below. The choice of $C_0$ is important in
practice, and is discussed in details in Section~\ref{design}.
\begin{theorem}
\label{main}
If $\pr(C)$ is generated from the randomization procedure in Definition \ref{procedure}, then the induced distribution of exposure $\pr(H)$ is $SCRD(\mathbf{n}_{A})$.
\end{theorem}
This result underpins the inferential approaches we describe in
Section~\ref{section:inference}. If the treatment intervention vector $W_0$
is degenerate, i.e it is a vector of $0$ values, then permuting $C_0$ is equivalent
to permuting $L_0$ only and Theorem \ref{main} reduces to Theorem 1
in \citet{basse2019randomization}.

\section{Inference}
\label{section:inference}

\subsection{Estimating the average exposure contrast}

Under Assumption~\ref{assumption}, our combined experiment can be thought of as a
multi-arm trials on the exposure scale. If the groups $L$ and treatment $W$ are
assigned according to Definition~\ref{procedure}, then Theorem~\ref{main} states
that this multi-arm trial follows a completely randomized design, stratified on
the attribute $A$.
Estimation and inference for average exposure contrast therefore follows immediately
from standard results in the randomization-based inference literature \cite{li2017general}. For
any $a \in \mathcal{A}$, and $k \in \mathcal{H}$, define

\begin{equation*}
  \hat{Y}_{[a]}(k) = \frac{1}{n_{[a](k)}} \sum_{i: A_i = a, H_i = k} Y_i,
\end{equation*}
the average outcome for units with attribute $A_i = a$ who receive the exposure $H_i = k$,
where $n_{[a]}(k) = |\{i \in \mathbb{I}: A_i = a, H_i = k\}|$. Consider
$\hat{\tau}_{[a]}(k, k') = \hat{Y}_{[a]}(k) - \hat{Y}_{[a]}(k')$ the difference-in-means
estimator within stratum $a$, and the stratified estimator
$\hat{\tau}(k, k') = \sum_{a\in\mathcal{A}} (n_{[a]}/ n) \hat{\tau}_{[a]}(k,k')$.
Theorem~\ref{asymptotics} summarizes their well-studied properties (see also Theorem~3 of
\citet{li2019randomization}).

\begin{theorem}
\label{asymptotics}
Under the randomization procedure in \ref{procedure}, and standard regularity conditions,
then for any $a \in \mathcal{A}$, $k, k' \in \mathcal{H}$, the estimators
$\hat{\tau}_{[a]}(k, k')$ and $\hat{\tau}(k, k')$ are unbiased for $\tau_{[a]}(k,k')$ and
$\tau(k,k')$ respectively, and are asymptotically normally distributed. In addition, the
standard Wald-type confidence interval for $\hat{\tau}_{[a]}(k,k')$ and
  $\hat{\tau}(k,k')$ are asymptotically conservative.
  
\end{theorem}

Stratified completely randomized designs also make it straightforward to incoporate
covariates in the analysis; see the Supplementary Material for details.

\subsection{Testing pairwise null hypotheses}

Building on recent literature on testing under interference \citep{Basse:2019b,aronow2012general,athey2018exact}, we construct a Fisher Randomization Test, conditioning on a focal set, defined as
\begin{equation*}
  \mathcal{U} = u(Z(L), W) = \{i \in \mathbb{I}: h_i(Z(L),W) \in \{(\gamma_1,w_1), (\gamma_2,w_2)\} \}.
\end{equation*}

Let the test statistic $T$ be the difference in means between the focal
units with exposure $(\gamma_1, w_1)$ and those with exposure $(\gamma_2,w_2)$.
The following proposition defines a valid test of $H_0^{(\gamma_1,w_1),(\gamma_2,w_2)}$. 

\begin{proposition}
\label{validity}
Consider observed $N-$vectors of exposure $H^{obs} \sim \pr(H)$ and outcome $Y^{obs} = Y(H^{obs})$, resulting in focal set $\mathcal{U}^{obs}$ and test statistic $T^{obs} = T(H^{obs}, Y^{obs}, \mathcal{U}^{obs})$. If $H' \sim \pr(H|\mathcal{U}^{obs})$ and $T' = T(H',Y^{obs}, \mathcal{U}^{obs})$, then the following quantity,
$$
\operatorname{pval}(H^{obs}) = \pr(T' \geq T^{obs}| \mathcal{U}^{obs})
$$
is a valid p-value conditionally and marginally for $H_0^{(\gamma_1,w_1),(\gamma_2,w_2)}$. That is, if $H_0^{(\gamma_1,w_1),(\gamma_2,w_2)}$ is true, then for any $\mathcal{U}^{obs}$ and $\alpha \in [0,1]$, we have $\pr\{\operatorname{pval}(H^{obs}) \leq \alpha | \mathcal{U}^{obs}\} \leq \alpha$.
\end{proposition}


Although it always leads to valid p-values, the test in Proposition~\ref{validity} is
computationally intractable for most choices of designs $pr(L, W)$. The challenge,
as highlighted by \cite{basse2019randomization}, is the step that requires sampling from the
conditional distribution of $pr(H \mid \mathcal{U}^{obs})$: even in small samples, this cannot
be accomplished by rejection sampling. Our key result in this section is that
if the design is symmetric in the sense of Section~\ref{randomization}, then the test in
Proposition~\ref{validity} can be carried efficiently:

\begin{theorem}
\label{main_2}
Let $\pr(C)$ be generated from randomization procedure described in Definition \ref{procedure} and $\pr(H)$ the induced exposure distribution. Define a focal set $\mathcal{U} = u(Z,W) = \{i \in \mathbb{I}: h_i(Z,W) \in \mathcal{H}_u\}$ for some $pr(W,Z) > 0$ and set of exposures $\mathcal{H}_u \subset \mathcal{H}$. Let $U = (U_1,\ldots, U_{N})$, where $U_i = \mathds{1}(i \in \mathcal{U})$. Then the conditional distribution of exposure, $pr(H| \mathcal{U})$, is $SCRD(\mathbf{n}_{A\mathcal{U}})$.
\end{theorem}

This theorem makes the test described in Proposition \ref{validity} computationally tractable by transforming
a difficult task --- sampling from an arbitrary conditional distribution --- into a simple
one --- sampling from a stratified completely randomized design.

\section{Optimal design heuristics}
\label{design}

Definition~\ref{procedure} requires the specification of an initial
pair $C_0 = (L_0, W_0)$. A straightforward consequence of
Theorem~\ref{main} is that the number of unit in a stratum $a$
receiving exposure $k$ is constant. Formally, let $H_0$ be the
exposure corresponding to $C_0$, and $H$ be any exposure vector that
may be induced by our procedure: we have
$n_{[a]k}(H_0) = n_{[a]k}(H)$, where
$n_{[a]k}(H') = |\{i \in \mathds{I}: H'_i = k, A_i = a\}|$.  If the
experimenter knows ex-ante that she is interested in estimating
$\tau(k, k')$, or testing the pairwise null $H_0^{k,k'}$, then a
useful heuristic for maximizing power would be to select $C_0$ such
that the associated exposure vector $H_0$ features many units with the
desired exposures $k$ and $k'$. Constructing such a $C_0$ manually is
possible in very small toy examples, but it becomes impractical as the
sample size increases even slightly. An alternative option would be to
perform a random search on the space of possible pairs $C = (L, W)$,
but it grows very fast as the number of clusters and their sizes
increases; making the process computationally challenging. Instead we
optimize our heuristic criterion directly.

Let $\mathcal{G} \subseteq (\mathcal{A} \times \mathcal{W})^m$ the set
of all possible attribute-intervention compositions for a group of
size $m$, so for any $G \in \mathcal{G}$,
$G = \{(a_1,w_1),\ldots,(a_m,w_m)\}$. For a group composition
$G \in \mathcal{G}$, target exposures $k, k' \in \mathcal{H}$, and
attribute $a \in \mathcal{A}$, let $m_k(G)$ and $c_a(G)$ be
respectively the number of units with exposure $k$ and the number of
units with attribute $a$, in group composition $G$. Finally, let
$n(G)$ be the number of groups with composition $G$.  Our heuristic
objective can formulated as the following integer linear program:
\begin{align*}
\operatorname{argmax}_{\{n_G\}_{G}} & \sum_{G \in \mathcal{G^*}} n(G) (m_{k}(G) + m_{k'}(G))\\
\operatorname{s.t.} & \sum_{G \in \mathcal{G^*}} n(G) c_a(G) \leq n_{[a]}, \forall a \in \mathcal{A}\\
& n(G)  \geq 0, n(G) \in \mathbb{Z}, \forall n(G).
\end{align*}
where $\mathcal{G}^\ast = \{G \in \mathcal{G}: m_k(G) + m_{k'}(G) > 0\}$. This
optimization problem can be solved efficiently numerically by relaxing the
integer constraint and rounding off the result. It does require enumerating the
set $\mathcal{G}^\ast$, but this is generally straightforward --- much more so
than enumerating the set of all possible assignment pairs. In particular,
$m_k(G), m_{k'}(G)$, and $c_a(G)$ can be computed for all
$G\in \mathcal{G}^\ast$ and all $a \in \mathcal{A}$, in constant time.

The objective criterion presented above seeks to maximize the number of units
receiving either exposure $k$ or exposure $k'$: this is a reasonable first order
criterion, but it has two drawbacks. First, it may lead to solutions with many
units exposed to $k$ or $k'$, but with a very unequal repartition: for instance,
we may have many units with exposure $k$, but none with exposure $k'$.  Smaller
imbalances may still have a large impact on the variance of stratified
estimators. Second, the number of units receiving each exposure may be balanced
overall but unbalanced within each stratum $a \in \mathcal{A}$, which may be
very problematic: indeed, we show in the Supplementary Material that in the
extreme case where all the units with exposure $k$ have attribute $a$ and all
the units with exposure $k'$ have exposure $a'$, our randomization test has no
power. Both issues can be addressed with minor modifications of the optimization
constraints presented above. We discuss the details and Supplementary Material, and
show that the resulting optimization problem is still an integer linear program.

\section{Simulation results}
We compare the power of our Procedure~\ref{randomization} for different design
strategies. We simulate a population of $N = 300$ units with binary attributes,
and consider a composite experiment that assigns these $N$ units to groups
of equal size $m$ for $m = 3, 4, 5, 6$, and then assigns a binary treatment
$W$ to a random subset of units. Using the exposure mapping of
Equation~\ref{eq:binary-exposure-mapping}, we focus on testing the null
hypothesis $H_0^{k, k'}$ where $k = (1, 1, 1, 1)$ and $k' = (2, 1, 1, 0)$. The
potential outcomes are generated as follows:
\begin{equation*}
  Y_i(k_0) = \begin{cases}
    M_i &\text{ if } k_0 \neq k'\\
    M_i + \tau &\text{ if } K_0 = k'
  \end{cases}
  \qquad \text{ where } \qquad
  M_i \sim \mathcal{N}(0, 1)
\end{equation*}
so that $H_0^{k,k'}$ holds for $\tau = 0$, and the magnitude of the violation
of the null is controlled by varying the parameter $\tau$ in the simulation.

\begin{figure}[!h]
\begin{minipage}[b]{0.5\textwidth}
\centering
\includegraphics[width=0.9\textwidth]{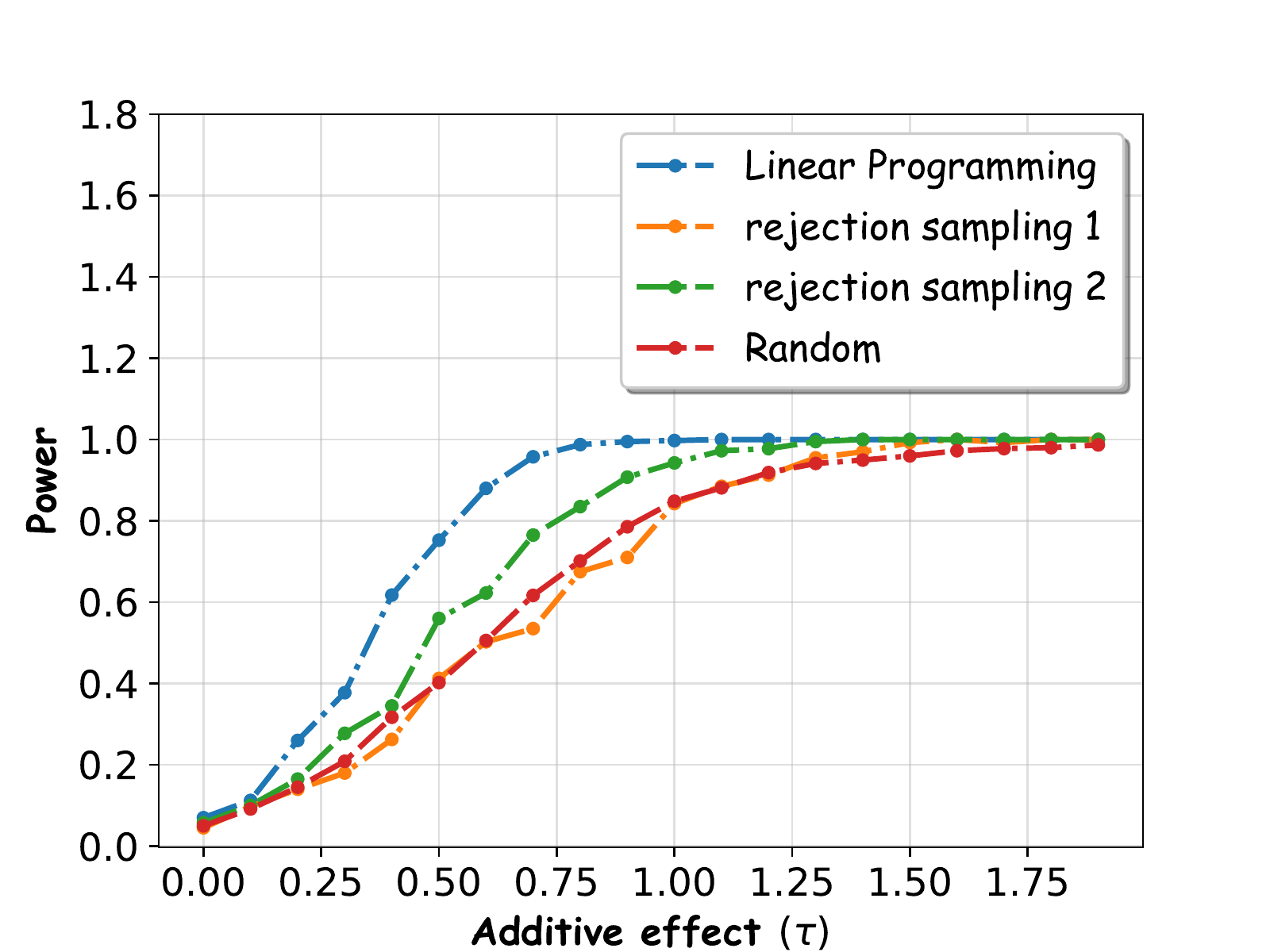}
\subcaption{$m=3$}
\label{m3}
\end{minipage}
\begin{minipage}[b]{0.5\textwidth}
\centering
\includegraphics[width=0.9\textwidth]{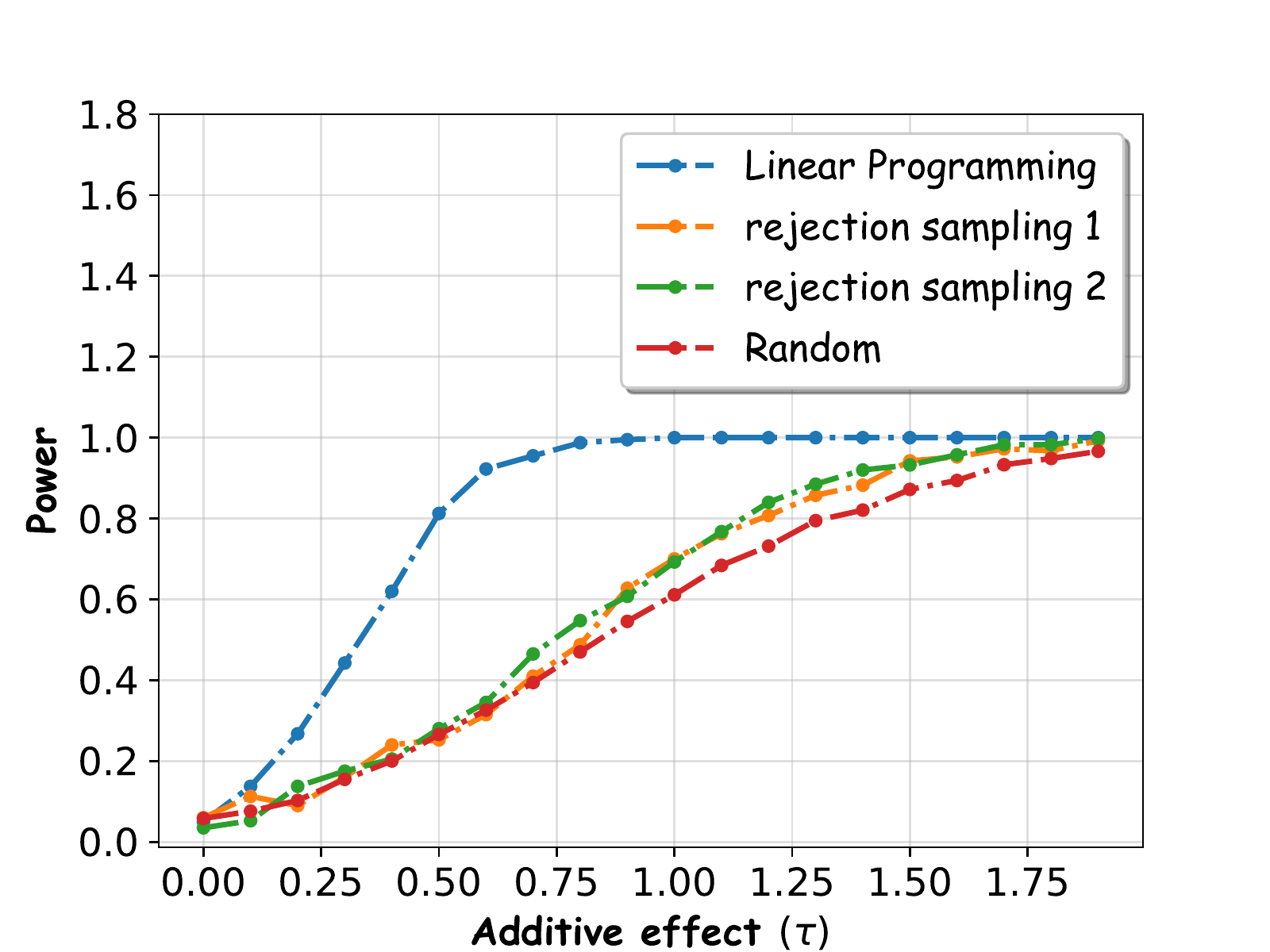}
\label{m4}
\subcaption{$m=4$}
\end{minipage}
\begin{minipage}[b]{0.5\textwidth}
\centering
\includegraphics[width=0.9\textwidth]{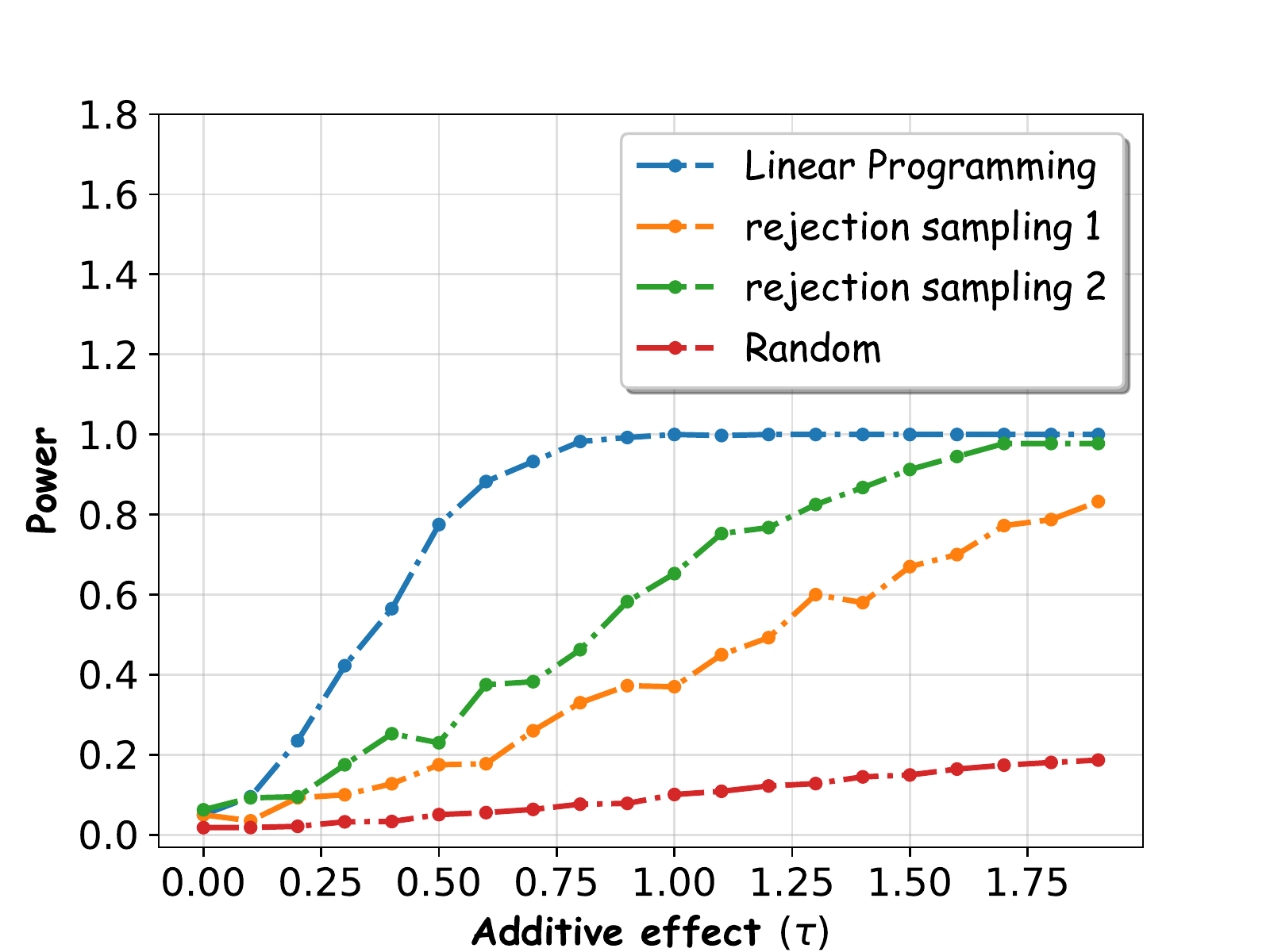}
\label{m5}
\subcaption{$m=5$}
\end{minipage}
\begin{minipage}[b]{0.5\textwidth}
\centering
\includegraphics[width=0.9\textwidth]{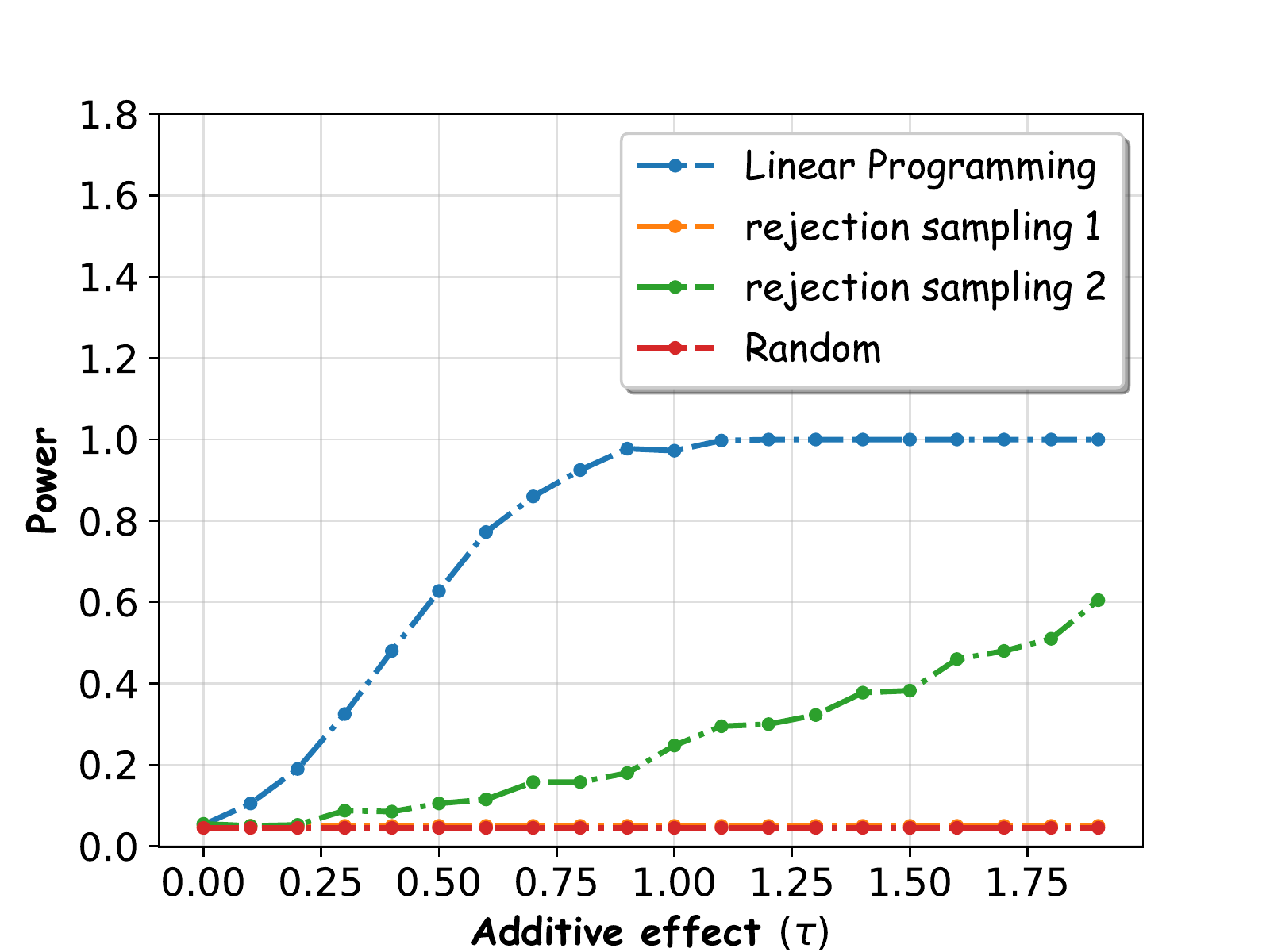}
\label{m6}
\subcaption{$m=6$}
\end{minipage}
\caption{Power against additive effect for different group sizes}
\label{fig:simul-figures}
\end{figure}

In all simulations, we use the randomization procedure described in
Definition~\ref{def:rand-procedure}, but we vary the choice of the initial
$C_0$ --- different choices lead to different designs. We compare the optimal
initialization strategy of Section~\ref{design} with three alternative
initialization strategies to assign $C$: 
\begin{enumerate}
\item Rejection sampling 1: The best initialization in $M  =10$ random permutations of the initial $C_0$
\item Rejection sampling 2:  The best initialization in $M  =1000$ random permutations of the initial $C_0$
\item Random initialization: A random permutation of the initial $C_0$
\end{enumerate}
More details
on the simulation setup can be found in the Supplementary Material. The results
of our simulations are plotted in Figure~\ref{fig:simul-figures}.

In our simulation, optimal design using linear programming leads to more
powerful tests than the other initializations for all additive effects and group
sizes we considered. The benefits of our linear programming strategy grow
starker as the size of the groups increases; indeed, for $m=6$, our optimal
design leads to tests that have a power of $1$ against the alternative $\tau=1$,
while the best alternative initialization strategy leads to tests of power less
than $0.3$. This is because as group size increases, the number of possible
exposures increases significantly and it is much more difficult for brute force
algorithm with a fixed number of iterations to find a near-optimal solution.



\clearpage
\singlespacing
\bibliographystyle{chicago}
\bibliography{ref}

\begin{thebibliography}{}

\bibitem[\protect\citeauthoryear{Aronow}{Aronow}{2012}]{aronow2012general}
Aronow, P.~M. (2012).
\newblock A general method for detecting interference between units in
  randomized experiments.
\newblock {\em Sociological Methods \& Research\/}~{\em 41\/}(1), 3--16.

\bibitem[\protect\citeauthoryear{Athey, Eckles, and Imbens}{Athey
  et~al.}{2018}]{athey2018exact}
Athey, S., D.~Eckles, and G.~W. Imbens (2018).
\newblock Exact p-values for network interference.
\newblock {\em Journal of the American Statistical Association\/}~{\em
  113\/}(521), 230--240.

\bibitem[\protect\citeauthoryear{Basse, Ding, Feller, and Toulis}{Basse
  et~al.}{2019}]{basse2019randomization}
Basse, G., P.~Ding, A.~Feller, and P.~Toulis (2019).
\newblock Randomization tests for peer effects in group formation experiments.
\newblock {\em arXiv preprint arXiv:1904.02308\/}.

\bibitem[\protect\citeauthoryear{Basse, Feller, and Toulis}{Basse
  et~al.}{2019}]{Basse:2019b}
Basse, G., A.~Feller, and P.~Toulis (2019).
\newblock Randomization tests of causal effects under interference.
\newblock {\em Biometrika\/}~{\em 106\/}(2), 487--494.

\bibitem[\protect\citeauthoryear{Cai and Szeidl}{Cai and
  Szeidl}{2018}]{cai2018interfirm}
Cai, J. and A.~Szeidl (2018).
\newblock Interfirm relationships and business performance.
\newblock {\em The Quarterly Journal of Economics\/}~{\em 133\/}(3),
  1229--1282.

\bibitem[\protect\citeauthoryear{Kimbrough, McGee, and Shigeoka}{Kimbrough
  et~al.}{2017}]{kimbrough2017peers}
Kimbrough, E.~O., A.~D. McGee, and H.~Shigeoka (2017).
\newblock How do peers impact learning? an experimental investigation of
  peer-to-peer teaching and ability tracking.
\newblock Technical report, National Bureau of Economic Research.

\bibitem[\protect\citeauthoryear{Li and Ding}{Li and
  Ding}{2017}]{li2017general}
Li, X. and P.~Ding (2017).
\newblock General forms of finite population central limit theorems with
  applications to causal inference.
\newblock {\em Journal of the American Statistical Association\/}~{\em
  112\/}(520), 1759--1769.

\bibitem[\protect\citeauthoryear{Li, Ding, Lin, Yang, and Liu}{Li
  et~al.}{2019}]{li2019randomization}
Li, X., P.~Ding, Q.~Lin, D.~Yang, and J.~S. Liu (2019).
\newblock Randomization inference for peer effects.
\newblock {\em Journal of the American Statistical Association\/}.

\bibitem[\protect\citeauthoryear{Sacerdote}{Sacerdote}{2001}]{sacerdote2001peer}
Sacerdote, B. (2001).
\newblock Peer effects with random assignment: Results for dartmouth roommates.
\newblock {\em The Quarterly journal of economics\/}~{\em 116\/}(2), 681--704.

\end{thebibliography}


\newpage
\appendix

\section{Proof of the main results}
\subsection{Elements of group theory}
Throughout this section, recall that $C = (C_1,\ldots, C_N) \in \mathbb{C}$, where $C_i = (L_i,W_i)$ is the pair of group assignment and additional intervention assignment. 

\begin{definition}[Group action on a set]
\label{group action}
Consider a permutation group $\Pi$ and a finite set of $N$-vector pairs, $\mathbb{C}$. A group action of $\Pi$ on $\mathbb{C}$ is a mapping $\phi: \Pi \times \mathbb{C} \to \mathbb{C}$ (usually we write $\pi \cdot C$ instead of $\phi(\pi, C)$) satisfying the following:
\begin{enumerate}
\item for all $C \in \mathbb{C}$, $e \cdot C = C$ where $e$ is the identity element of $\Pi$;
\item for all $\pi,\pi' \in \Pi$, and all $C \in \mathbb{C}$, $\pi' \cdot(\pi \cdot C) = (\pi' \pi) \cdot C$
\end{enumerate}
\end{definition}

It can be checked that for $\pi \in S$ and $C = (C_1,\ldots,C_{N}) \in \mathbb{C}$, the mapping $\pi \cdot C = (C_{\pi^{-1}(i)})_{i=1}^N = (L_{\pi^{-1}(i)}, W_{\pi^{-1}(i)})_{i=1}^N$ is a group action.

\begin{definition}[Orbits and stabilizers]
Let $\Pi$ be a permutation group and $\mathbb{C}$ a finite set of $N$-vectors. If $C \in \mathbb{C}$, the orbit of $C$ under $\Pi$ is defined as
$$
\Pi \cdot C \equiv \{\pi \cdot C: \pi \in \Pi\},
$$
and the stabilizer of $C$ in $\Pi$ is defined as 
$$
\Pi_C \equiv \{\pi \in \Pi: \pi \cdot C = C\}.
$$
\end{definition}

Recall the definition of a transitive group action in the main text.
\begin{definition}[Transitivity]
A subgroup $\Pi \subset S$ of the symmetric group $S$ acts transitively on $\mathbb{C}$ if $\mathbb{C} = \Pi \cdot C$ for any $C \in \mathbb{C}$. 
\end{definition} 

We will now state a version of the Orbit-Stabilizer Theorem that will is specific to our setup. 
\begin{theorem}[Orbit-Stabilizer]
Let $\Pi$ be a permutation group acting transitively on a finite set of $N$-vectors $\mathbb{C}$.
\begin{enumerate}
\item For all $C, C' \in \mathbb{C}$, $|\Pi_C| = |\Pi_{C'}| = D$ a constant. In words, it means that all stabilizers have the same size. 
\item We already know that for all $C \in \mathbb{C}$, $\Pi \cdot C = \mathbb{C}$. We also have:
$$
|\Pi \cdot C| = \frac{|\Pi|}{|\Pi_C|} = \frac{|\Pi|}{D}. 
$$
\end{enumerate}
\end{theorem}

\subsection{Proof of Theorem \ref{main}}
\vspace{5mm}\noindent{\bf Theorem 1}. {\it 
If $\pr(C)$ is generated from the randomization procedure in Definition \ref{procedure}, then the induced distribution of exposure $\pr(H)$ is $SCRD(\mathbf{n}_{A})$.}\vspace{5mm}

The proof for Theorem \ref{main} can be split into two parts. The first part is about showing equivariance of exposure mapping under permutation of latent assignments, and the second part is about establishing symmetry property.

\begin{lemma}
\label{equivariance}
Let $\Pi$ be a subgroup of $S_A$, the stabilizer of the attribute vector $A$ in $S$. For $C \in \mathbb{C} \subset \{1,\ldots,K\}^{N} \times \mathcal{W}^{N}$, define $h^*(C)=h(W,Z(L))$, where $h_i(W,Z(L))=\{W_i\} \cup \{(A_j,W_j), j \in Z_i\}\}$ is the exposure mapping of unit $i$ and domain $\mathbb{H} = \{h^*(C): C \in \mathbb{C}\}$. Then we have that $h^*: \mathbb{C} \to \mathbb{H}$ is equivariant with respect to $\Pi$.
\end{lemma}

\begin{proof}
 We will show that $h^*(\pi \cdot C) = \pi \cdot  h^*(C)$ for all $C \in \mathbb{C}$ and all $\pi \in \Pi$. 

Consider a fixed $C \in \mathbb{C}$ and $\pi \in \Pi$. By definition, we have
$$
[h^*(C)]_i = \{W_i\} \cup \{(A_j,W_j): j \neq i, L_j = L_i\}. 
$$
Then we have for all $i \in \mathbb{I}$,
\begin{align*} 
[h^*(\pi \cdot C)]_i &= \{[\pi \cdot W]_i\} \cup \{([\pi \cdot A]_j,[\pi \cdot W]_j): j \neq i, [\pi \cdot L]_j = [\pi \cdot L]_i\} \\
&= \{W_{\pi^{-1}(i)}\} \cup \{(A_{\pi^{-1}(j)},W_{\pi^{-1}(j)}): j \neq i, L_{\pi^{-1}(j)}= L_{\pi^{-1}(i)}\} \\
&= \{W_{\pi^{-1}(i)}\} \cup \{(A_j,W_{\pi^{-1}(j)}): j \neq i, L_{\pi^{-1}(j)}= L_{\pi^{-1}(i)}\} \\
&= \{W_{\pi^{-1}(i)}\} \cup \{(A_{\pi(\pi^{-1}(j))},W_{\pi^{-1}(j)}): \pi(\pi^{-1}(j)) \neq i, L_{\pi^{-1}(j)}= L_{\pi^{-1}(i)}\}\\
&= \{W_{\pi^{-1}(i)}\} \cup \{(A_{\pi(j')},W_{j'}): j' \neq \pi^{-1}(i), L_{j'} = L_{\pi^{-1}(i)} \} \\
&= \{W_{\pi^{-1}(i)}\} \cup \{(A_{j'},W_{j'}): j' \neq \pi^{-1}(i), L_{j'} = L_{\pi^{-1}(i)} \} \\
&= [h^*(C)]_{\pi^{-1}(i)} \\
&= [\pi \cdot h^*(C)]_i
\end{align*}
\end{proof}

Lemma \ref{equivariance} shows that exposure mapping is equivariant with respect to simultaneous permutation of the group and external intervention treatment assignments. In other words, permuting the latent assignment vector $C$ is equivalent to permuting the exposure mappings. This allows symmetry properties to propagate from latent assignments to the induced exposure distribution. Specifically, we focus on designs for which the exposure has a Stratified Completely Randomized Design. Recall the notion of $\operatorname{SCRD}(\mathbf{n}_{A})$ in Definition \ref{SCRD}.
 
\vspace{5 mm}\noindent{\bf Definition 1}. {\it 
Without loss of generality, denote by $\mathcal{H}$ the set of possible exposures and $A$ an N-vector. Let $\mathbf{n}_A = (n_{a,h})_{a \in \mathcal{A}, h \in \mathcal{H}}$, such that $\sum_{h\in \mathcal{H}} n_{a,h} = N_{[a]}$, denote a vector of non-negative integers corresponding to number of units with each possible attribute and exposure combination. We say that a distribution of exposures $\pr(H)$ is a stratified completely randomized design denoted by $\operatorname{SCRD}(\mathbf{n}_{A})$ if the following two conditions are satisfied. 
\begin{enumerate}
\item After stratifying based on $A$, the exposure $H = (H_1,\ldots, H_N)$ is completely randomized. That is, (1) $\mathbb{P}(H_i = h) = \mathbb{P}(H_j = h)$ for all $h \in \mathcal{H}$ and $i,j \in \mathbb{I}$ such that $A_i = A_j$; (2) the number of units with exposure $h \in \mathcal{H}$ and stratum $a \in \mathcal{A}$ is $n_{a,h}$.
\item The exposure assignments across strata are independent. That is $\mathbb{P}(H_i = h_i | H_j = h_j) = \mathbb{P}(H_i = h_i)$ for all $h_i, h_j \in \mathcal{H}$ and $i,j \in \mathbb{I}$ such that $A_i \neq A_j$. 
\end{enumerate}}\vspace{5mm}

\begin{lemma}
\label{SCRD_proof}
Fix any $H_0 \in \mathbb{H} = \{h^*(C): C \in \mathbb{C}\}$ and generate $H = \pi \cdot H_0$ where $\pi \sim \operatorname{Unif}(S_A)$. Then the distribution of exposures $\pr(H)$ is $\operatorname{SCRD}(\mathbf{n}_{A})$.
\end{lemma}

\begin{proof}
We first note that if we permute $H$ by $\pi \in \operatorname{Unif}(S)$, then $\pr(H)$ is completely randomized (CRD). This is because with a random permutation, $\mathbb{P}(H_i = h) = \mathbb{P}(H_j = h)$ for all $h \in \mathcal{H}$. We then proceed by proving the two conditions in the definition for $\operatorname{SCRD}(\mathbf{n}_{A})$ separately. 
\begin{enumerate}
\item We will show that $\pr(H)$ satisfies completely randomized design (CRD) within each stratum defined by attribute vector $A$, i.e. $\mathbb{P}(H_i = h) = \mathbb{P}(H_j = h)$ for all $h \in \mathcal{H}$ and $i,j \in \mathbb{I}$ such that $A_i = A_j$. 

\noindent For each stratum $k$ as defined from $A$, let 
$$
I_k \equiv \{i \in \mathbb{I}: A_i = k\}.
$$ 
For $\pi \in S_A$, let $\tilde{\pi}^{(k)}: I_k \mapsto \mathbb{I}$ be the restriction of $\pi$ to $I_k$ such that $\tilde{\pi}^{(k)}(i) = \pi(i)$. Since $\pi \in S_A$, $\pi(i) \in I_k$, $\forall i \in I_k$. Therefore $\operatorname{Img}(\tilde{\pi}^{(k)}) \subset I_k$. But since $\pi$ is a permutation, $\tilde{\pi}^{(k)}$ is a bijection. Therefore $\operatorname{Img}(\tilde{\pi}^{(k)}) = I_k$. This shows that $\tilde{\pi}^{(k)} \in S^k$ where $S^k$ is the symmetric group on $I_k$. 

\noindent We then characterize the induced distribution of $\tilde{\pi}^{(k)}$ on $S^k$, where we sample $\pi \sim \operatorname{Unif}(S_A)$. Define the following $N$-vector $L^k = (L_1^k,\ldots, L_N^k)$ where 
$$
L_i^k = 
\begin{cases}
0, &\text{ if } i \notin I_k\\
i, &\text{ if } i \in I_k.
\end{cases}
$$

\noindent For any $\pi^* \in S^k$, we have 
\begin{align*}
\pr(\pi^*) &= \sum_{\pi \in S_A} 1\{\tilde{\pi}^{(k)} = \pi^*\}\frac{1}{|S_A|} \\
&= \frac{|S_{AL^k}|}{|S_A|} \\
&= \frac{1}{|S_A \cdot L^k|},
\end{align*}
where the last line is due to the Orbit-Stabilizer Theorem. We will further show that $|S_A \cdot L^k| = |S^k|$. 

\noindent For any $\pi \in S_A$, $(\pi \cdot L^k)_i \notin I_k$, $\forall i \notin I_k$. By the definition of $L^k$, we know that $(\pi \cdot L^k)_i = L_i^k = 0$, $\forall i \notin I_k$. Therefore 
$$
|S_A \cdot L^k| \leq |I_k|! = |S^k|.
$$

\noindent For the opposite inequality, consider any permutations $\pi \in S_k$ acting on $L^k$ restricted to $I_k$. Define the extended permutation $\tilde{\pi}$ on $\mathbb{I}$ by 
$$
\tilde{\pi}(i) = 
\begin{cases}
i, & i \notin I_k\\
\pi(i) & i \in I_k,
\end{cases}
$$
and denote the set of all such $\tilde{\pi}$ as $\tilde{S}^k$. Since $L_i^k \neq L_j^k$, $\forall i \neq j$ and $i,j \in I_k$, $|\tilde{S}^k \cdot L^k| = |S^k|$. Since by construction, $\tilde{S}^k \subset S_A$, we have that 
$$
|S_A \cdot L^k| \geq |\tilde{S}^k \cdot L^k| = |S^k|.
$$
Combining the two inequalities together, we have 
$$
\pr(\pi^*) =\frac{1}{|S_A \cdot L^k|} = \frac{1}{|S^k|}.
$$
This implies that the induced restricted permutations $\tilde{\pi}^{(k)} \sim \operatorname{Unif}(S^k)$ for all $k$. In other words, $\pr(H)$ satisfies CRD within each stratum $k$, and hence $\mathbb{P}(H_i = h) = \mathbb{P}(H_j = h)$ for all $h \in \mathcal{H}$ and $i,j \in \mathbb{I}$ such that $A_i = A_j$.
\item We will show that exposure assignments are independent across strata. First notice that for all $h_i, h_j \in \mathcal{H}$ and $i,j \in \mathbb{I}$ such that $A_i \neq A_j$,
\begin{align*}
\mathbb{P}(H_i = h_i \vert H_j = h_j) &=\mathbb{P}((\pi\cdot H_0)_i = h_i \vert (\pi \cdot H_0)_j = h_j)\\
&= \mathbb{P}(\pi(i) \in D | \pi(j) \in E),
\end{align*}
for some disjoint sets $D,E \subset \mathbb{I}$ such that $D \subset I_{A_i}$ and $E \subset I_{A_j}$. By Baye's rule we have,
\begin{align*}
\mathbb{P}(\pi(i) \in D \vert \pi(j) \in E) &= \frac{\mathbb{P}(\pi(i) \in D, \pi(j) \in E)}{\mathbb{P}(\pi(j) \in E)} \\
&=  \frac{\mathbb{P}(\pi(i) \in D, \pi(j) \in E)}{\sum_{i^* \in I_{A_i}} \mathbb{P}(\pi(i)= i^*, \pi(j) \in E)}\\
&= \frac{|D|}{|I_{A_i}|},
\end{align*}
where last equality is because $\mathbb{P}(\pi(i)= i^*, \pi(j) \in E)$ is the same for all $i^* \in I_{A_i}$. 
Finally we have 
$$
\mathbb{P}(H_i = h_i \vert H_j = h_j) = \frac{|D|}{|I_{A_i}|} = \mathbb{P}(\pi(i) \in D) = \mathbb{P}(H_i = h_i), 
$$
where the second equality is due to CRD within $I_{A_i}$ in part $(1)$.
\end{enumerate}
\end{proof}
Combining the above two Lemmas together proves Theorem \ref{main} that $\pr(H)$ is $\operatorname{SCRD}(\textbf{n}_A)$.

\subsection{Proof of Proposition \ref{validity}}

\vspace{5 mm}\noindent{\bf Proposition 1}. {\it 
Consider observed $N-$vectors of exposure $H^{obs} \sim \pr(H)$ and outcome $Y^{obs} = Y(H^{obs})$, resulting in focal set $\mathcal{U}^{obs}$ and test statistic $T^{obs} = T(H^{obs}, Y^{obs}, \mathcal{U}^{obs})$. If $H' \sim \pr(H|\mathcal{U}^{obs})$ and $T' = T(H',Y^{obs}, \mathcal{U}^{obs})$, then the following quantity,
$$
\operatorname{pval}(H^{obs}) = \pr(T' \geq T^{obs}| \mathcal{U}^{obs})
$$
is a valid p-value conditionally and marginally for $H_0^{(\gamma_1,w_1),(\gamma_2,w_2)}$. That is, if $H_0^{(\gamma_1,w_1),(\gamma_2,w_2)}$ is true, then for any $\mathcal{U}^{obs}$ and $\alpha \in [0,1]$, we have $\pr\{\operatorname{pval}(H^{obs}) \leq \alpha | \mathcal{U}^{obs}\} \leq \alpha$.
}\vspace{5mm}

\begin{proof}
Recall that $u(Z(L), W) = \{i \in \mathbb{I}: (h_i(Z(L),W),W_i) \in  \{(\gamma_1,w_1), (\gamma_2,w_2)\} \}$. Define $m(\mathcal{U}|C) = \mathds{1}\{u(Z(L),W) = \mathcal{U}\}$. Then,
\begin{equation*}
m(\mathcal{U}|C) >0 \Rightarrow u(Z(L),W) = \mathcal{U} \Rightarrow h_i(Z(L), W) \in \{(\gamma_1,w_1), (\gamma_2,w_2)\}, \forall i \in \mathcal{U}.
\end{equation*} 
Therefore $\pr( C | \mathcal{U}) > 0$ implies that $\pr(\mathcal{U}|C) >0$, and hence \\
$h_i(Z(L), W) \in \{(\gamma_1,w_1), (\gamma_2,w_2)\}, \forall i \in \mathcal{U}$. For all $C,C'$ such that $\pr(C| \mathcal{U}) >0 $ and $\pr(C'|\mathcal{U}) > 0$, we must have $h_i(Z(L), W), h_i(Z(L'), W') \in \{(\gamma_1,w_1), (\gamma_2,w_2)\}, \forall i \in \mathcal{U}$. Therefore under the null hypothesis $H_0^{(\gamma_1,w_1),(\gamma_2,w_2)}$, $Y_i(Z(L),W) = Y_i(Z(L'),W') = Y_i(\gamma_1,w_1) = Y_i (\gamma_2,w_2)$. This means that under $H_0^{(\gamma_1,w_1),(\gamma_2,w_2)}$, the test statistic $T$ is imputable. The result then follows from Theorem 2.1 of \cite{basse2019randomization}.
\end{proof}

\subsection{Proof of Theorem \ref{main_2}}

\vspace{5mm}\noindent{\bf Theorem 3}. {\it 
Let $\pr(C)$ be generated from randomization procedure described in Definition \ref{procedure} and $\pr(H)$ the induced exposure distribution. Define a focal set $\mathcal{U} = u(Z,W) = \{i \in \mathbb{I}: h_i(Z,W) \in \mathcal{H}_u\}$ for some $pr(W,Z) > 0$ and set of exposures $\mathcal{H}_u \subset \mathcal{H}$. Let $U = (U_1,\ldots, U_{N})$, where $U_i = \mathds{1}(i \in \mathcal{U})$. Then the conditional distribution of exposure, $pr(H| \mathcal{U})$, is $SCRD(\mathbf{n}_{A\mathcal{U}})$.}\vspace{5mm}

In order to prove Theorem \ref{main_2}, we need to introduce concepts of group symmetry and then establish the connection between group symmetry and $SCRD(\mathbf{n}_{A\mathcal{U}})$.
\begin{definition}[$\Pi$-symmetry]
Let $\Pi \subset S$ be a subgroup of the symmetric group $S$. A distribution, $\pr(H)$ with domain $\mathbb{H}$ is called $\Pi$-symmetric if $\pr(H) = \operatorname{Unif}(\mathbb{H})$ and $\Pi$ acts transitively on $\mathbb{H}$.
\end{definition}

The following Proposition establishes connections between $\Pi$-symmetry and sampling procedure.
\begin{proposition}
\label{link}
Let $\Pi \subset S$ be a subgroup of the symmetric group $S$ and $\mathbb{H} = \{h^*(C): C \in \mathbb{C}\}$. Take any $H_0 \in \mathbb{H}$ and define $\mathbb{H}_0 = \Pi \cdot H_0$.
\begin{enumerate}
\item If we sample $H = \pi \cdot H_0$, where $\pi \sim \operatorname{Unif}(\Pi)$, then the distribution of $H$ is $\Pi$-symmetric on $\mathbb{H}_0$.
\item If a distribution of $H$ is $\Pi$-symmetric on its domain $\mathbb{H}_0$, then it can be generated by sampling $H = \pi \cdot H_0$, where $\pi \sim \operatorname{Unif}(\Pi)$.
\end{enumerate}
\end{proposition}
\begin{proof}
The proof for part (1) and part (2) are identical. 
The definition of $\Pi$-symmetry involves two parts, namely transitivity and uniform distribution on the support. 
We first show that $\Pi$ acts transitively on the set $\mathbb{H}_0$, that is for all $H \in \mathbb{H}_0$, $\Pi \cdot H = \mathbb{H}_0$. 

By construction, for all $H \in \mathbb{H}_0$, there exists $\pi_0 \in \Pi$ such that $H = \pi_0 \cdot H_0$. Therefore transitivity condition of $\Pi \cdot H = \mathbb{H}_0$ can also be written as 
$$
(\Pi \pi_0) \cdot H_0 = \Pi \cdot H_0.
$$
 To prove transitivity, it then suffices to show that $\Pi \pi_0 = \Pi$. 

Since for all $\pi \in \Pi \pi_0$, there exists $\pi' \in \Pi$ such that $\pi = \pi' \pi_0 \in \Pi$, we have $\Pi \pi_0 \subset \Pi$. For the reverse direction, consider $\pi \in \Pi$, we can expand $\pi = \pi \pi_0^{-1} \pi_0 \in \Pi \pi_0$ since $\pi \pi_0^{-1} \in \Pi$. Therefore $\Pi\pi_0 = \Pi$ and hence transitivity holds.

Before moving on to the second part, we first clarify some notations. Define $\pr_\Pi(\pi)  = \operatorname{Unif}(\Pi)$ and $\pr_{\Pi,H_0}(H)$ the distribution of $H$ generated by the sampling procedure: that is, the distribution of $H$ obtained by first sampling $\pi$ from $\pr_\Pi(\pi)$ and then applying $\pi \cdot H_0$. It remains to prove that $\pr_{\Pi,H_0}(H) = \operatorname{Unif}(\mathbb{H}_0)$.

Again we have for any $H \in \mathbb{H}_0$, there exists $\pi_0 \in \Pi$ such that $H = \pi_0 \cdot H_0$. This means that $H_0  = \pi_0^{-1}\cdot H$ for some $\pi_0^{-1} \in \Pi$. Therefore 
\begin{align*}
\pr_{\Pi,H_0}(H) & = \sum_{\pi \in \Pi} \mathds{1}(\pi \cdot H_0 = H) \pr_\Pi(\pi)\\
&= \sum_{\pi \in \Pi} \mathds{1}(\pi \cdot (\pi_0^{-1} \cdot H) = H)\pr_\Pi(\pi) \\
&= \sum_{\pi \in \Pi} \mathds{1}((\pi \pi_0^{-1}) \cdot H = H)\pr_\Pi(\pi) \\
&= \sum_{\pi \in \Pi} \mathds{1}(\pi \pi_0^{-1} \in \Pi_H) \pr_\Pi(\pi) \\
&= \sum_{\pi \in \Pi} \mathds{1}(\pi \in \Pi_H \pi_0)\pr_\Pi(\pi) \\
&= \pr_\Pi(\Pi_H \pi_0),
\end{align*}
where $\Pi_H$ is the stabilizer of $H$ in $\Pi$. Since $\pr_\Pi(\pi) = \operatorname{Unif}(\Pi)$ and $\Pi_H \pi_0 \subset \Pi$, we have 
\begin{equation}
\label{symmetry_1}
\pr_{\Pi, H_0}(H) = \pr_\Pi(\Pi_H \pi_0) = \frac{|\Pi_H \pi_0|}{|\Pi|}.
\end{equation}
We quickly verify that $|\Pi_H \pi_0 | = |\Pi_H|$.
Clearly, $|\Pi_H\pi_0| \leq |\Pi_H|$ and we only need to verify the other direction. Suppose that there exist $\pi_1, \pi_2 \in \Pi_H$ such that $\pi_1 \neq \pi_2$ but $\pi_1\pi_0 = \pi_2 \pi_0$. Then this would imply
\begin{equation*}
\pi_1\pi_0 \pi_0^{-1} = \pi_2 \Longrightarrow \pi_1 = \pi_2,
\end{equation*}
which is a contradiction. Since $\pi_1 \neq \pi_2$ implies $\pi_1 \pi_0 \neq \pi_2 \pi_0$, we know that $|\Pi_H\pi_0| \geq |\Pi_H|$. Therefore $|\Pi_H\pi_0| = |\Pi_H|$. Applying this to Equation \eqref{symmetry_1}, we have 
\begin{equation*}
\pr_{\Pi, H_0}(H) = \frac{|\Pi_H|}{|\Pi|}.
\end{equation*}
By the Orbit-Stabilizer Theorem, $|\Pi \cdot H| = |\Pi| / |\Pi_H|$. Therefore 
\begin{equation*}
\pr_{\Pi, H_0}(H) = |\Pi \cdot H|^{-1} = |\mathbb{H}_0|^{-1} = \operatorname{Unif}(\mathbb{H}_0),
\end{equation*}
where the second equality is due to transitivity that we proved earlier. Therefore $\pr_{\Pi, H_0}(H)$ is $\Pi$-symmetric on $\mathbb{H}_0$.
\end{proof}

We now proceed to prove Theorem \ref{main_2} in two steps. The first step tries to characterize symmetry property of $\pr(H|\mathcal{U})$ and the second step relates symmetry property to $SCRD(\mathbf{n}_{A\mathcal{U}})$.
\begin{proposition}
Let $\pr(C)$ be generated from randomization procedure in Definition \ref{procedure} and $\pr(H)$ the induced exposure distribution. Define a focal set $\mathcal{U} = u(Z,W) = \{i \in \mathbb{I}: h_i(Z,W) \in \mathcal{H}_u\}$ for some $pr(W,Z) > 0$ and set of exposures $\mathcal{H}_u \subset \mathcal{H}$. Let $U = (U_1,\ldots, U_{N})$, where $U_i = \mathds{1}(i \in \mathcal{U})$. Then the conditional distribution of exposure, $pr(H| \mathcal{U})$, is $S_{AU}-$symmetric, where $S_{AU}$ is the stabilizer of both $A$ and $U$ in $S$.
\end{proposition}
\begin{proof}
First recall that due to equivariance in Lemma \ref{equivariance}, the induced $\pr(H)$ is generated by sampling $H = \pi \cdot H_0$, where $\pi \in S_A$. By first part of Proposition \ref{link}, we know that the distribution of exposures $\pr(H)$ is $S_A$-symmetric on its domain $\mathbb{H}$. In particular, it has a uniform distribution on $\mathbb{H}$. 

Notice that the function $u(\cdot)$ depends on $(W,Z)$ only through $H = h(W,Z)$. This makes it possible to define another function $m(\cdot)$ such that $\mathcal{U} = m(H) = m(h(W,Z)) = u(W,Z)$. Since there is a one-to-one mapping between $\mathcal{U}$ and $U$, we can use the two notations interchangeably. The reason that $U$ is a useful representation is that it is an $N-$vector, allowing previous notations of permutation to work out. Here we can write $U = m(H)$.

We have
\begin{align*}
pr(H | \mathcal{U}) &\propto pr(\mathcal{U} |H)pr(H)\\
&\propto pr(\mathcal{U}|H)  \text{ since $pr(H) = \operatorname{Unif}(\mathbb{H}) \propto 1$}\\
&\propto \mathds{1}\{m(H) = U\},
\end{align*}
which implies that $pr(H|\mathcal{U}) = \operatorname{Unif}\{\mathbb{H}(U)\}$ on the support 
\begin{equation*}
\mathbb{H}(U) = \{H \in \mathbb{H}: m(H) = U\}.
\end{equation*}
Now notice that for all $\pi \in S_A$ and any exposure set of interest $h_1$, $h_2$, we have
\begin{align*}
[m(\pi \cdot H)]_i &= \mathds{1}\left([\pi \cdot H]_i \in \mathcal{H}_u\right)\\
&=  \mathds{1}\left(H_{\pi^{-1}(i)} \in \mathcal{H}_u \right) \\
&= [m(H)]_{\pi^{-1}(i)}\\
&= [\pi \cdot m(H)]_i,
\end{align*}
that is, $m$ is equivariant, $m(\pi \cdot H) = \pi \cdot m(H)$. Let $H_0 \in \mathbb{H}(U)$ such that $m(H_0) = U$. We have,
\begin{align*}
\mathbb{H}(U) &= \{H \in \mathbb{H}: m(H)= U\} \\
&= \{\pi \cdot H_0: \pi \in S_A, m(\pi \cdot H_0) = U\} ,\text{ since $S_A$ is transitive on $\mathbb{H}$}\\
&= \{\pi \cdot H_0: \pi \in S_A, \pi \cdot m(H_0) = U\}, \text{ due to equivariance}\\
&= \{\pi \cdot H_0: \pi \in S_A, \pi \cdot U = U\}\\
&= S_{AU} \cdot H_0.
\end{align*}
This shows that $S_{AU}$ is transitive on $\mathbb{H}(U)$, the support of $\pr(H|\mathcal{U})$. Having shown earlier that $\pr(W |\mathcal{U}) = \operatorname{Unif}\{\mathbb{H}(U)\}$, we therefore conclude that $\pr(H|\mathcal{U})$ is $S_{AU}$-symmetric on its support. 
\end{proof}

Since $\pr(H|\mathcal{U})$ is $S_{AU}$-symmetric on its support $\mathbb{H}(U)$, we know by part (2) of Proposition \ref{link}, $\pr(H|\mathcal{U})$ can be generated by sampling $H = \pi \cdot H_0$, where $\pi \sim \operatorname{Unif}(S_{AU})$. Then the second step is to invoke Lemma \ref{SCRD_proof} in proof of Theorem \ref{main} except that we replace the $N$-vector $A$ with $AU$ where $(AU)_i = (A_i, U_i)$. This completes the proof that $\pr(H|\mathcal{U})$ is $SCRD(\mathbf{n}_{A\mathcal{U}})$.

\section{Testing for sharp null hypothesis}
Consider testing the global null hypothesis
\begin{equation*}
H_0: Y_i(\gamma_1,w_1) = Y_i(\gamma_2,w_2), \forall (\gamma_1,w_1),(\gamma_2,w_2) \in \mathcal{H}, \forall i \in \mathbb{I},
\end{equation*}
which asserts that the combined intervention has no effect whatsoever on any unit. We illustrate here how the classical Fisher Randomization Test can be applied to test this sharp null hypothesis. 
\begin{proposition}
Consider observed assignment $C^{obs} \sim \pr(C)$.
\begin{enumerate}
\item Observe outcomes, $Y^{obs}= Y(C^{obs}) = Y(Z(L^{obs}), W^{obs})$, where $C^{obs}_i = (L^{obs}_i, W^{obs}_i)$ for all $i \in \mathbb{I}$.
\item Compute $T^{obs} = T(C^{obs}, Y^{obs})$.
\item For $C' \sim \pr(C)$, let $T' = T(C',Y^{obs})$ and define:
$$
\operatorname{pval}(C^{obs}) = \pr(T' \geq T^{obs}),
$$
where $T^{obs}$ is fixed and the randomization distribution is with respect to $\pr(Z')$.
\end{enumerate}
Then the p-value of $\operatorname{pval}(C^{obs})$ is valid. That is, if $H_0$ is true, then $\pr\{\operatorname{pval}(C^{obs}) \leq \alpha\} \leq \alpha$.
\end{proposition}

\section{Balance in optimal design heuristics}
 The naive approach in section \ref{design} only considers the objective of maximizing the total number of units with both target exposures, without requiring balance between the two exposures. We will show how to reformulate the optimization to incorporate balance by adding various constraints. But before that, we want to point out the subtleties in incorporating balance as well as the caveats in incorporating balance in the wrong way or simply ignoring it. 

 Recall that the randomizations in Definition \ref{procedure} are permutations that are in the stabilizer of attribute $A$. This suggests that balance between the two target exposures should be taken into consideration within each category of attribute instead of on the global level across all attribute values. In fact, considering balance between the two target exposures without taking into account of diversity within each attribute class could result in greedy choice that leads to zero power. For example, if all units with the first target exposure are of attribute $a_1 \in \mathcal{A}$ while all units with the second target exposure are of attribute $a_2 \in \mathcal{A}$ for $a_1 \neq a_2$, then permutations in the stabilizer of $A$ do not change the test statistics at all. In this worst case, we will have zero power. Similarly, in the naive approach that neglects the balance between the two target exposures, the same worst case scenarios may happen resulting in zero power. 

 It is worth noting that the correct way to incorporate balance and the heuristics for maximizing power of randomization tests also coincide with the goal of minimizing variance in estimations. From standard theory about estimation of variance, it can be seen that variance estimator is small if the denominators $n_{[a]\gamma,w}$ and $n_{[a]\gamma',w'}$ are large for both target exposures $(\gamma,w)$ and $(\gamma',w')$ within attribute class $a \in \mathcal{A}$. 	This suggests that an optimal design desires large values of both $n_{[a]\gamma,w}$ and $n_{[a]\gamma',w'}$, which can be implemented by maximizing the sum of units with both target exposures, subject to the within-attribute balance constraints. We will now formally state the reformulation of the integer linear programming problem.

 Given target exposures $(\gamma_1,w_1)$ and $(\gamma_2,w_2)$, we know the exact composition of attribute-intervention pair of the neighbors of all units with target treatments. This allows us to enumerate all elements in $\mathcal{G^*}$ and hence pre-compute the constants $m_1(G_i)$, $m_2(G_i)$, and $c_j(G_i)$ for all $G_i \in \mathcal{G^*}$ and $j \in \mathcal{A}$. 

 Assume without loss of generality that $\mathcal{A} = \{0,1\}$. Define the following additional constants
 \begin{equation*}
 A_1(G_i) = 
 \begin{cases}
 1, & \text{if the units with exposure $1$ in group design $G_i$}\\ 
 & \text{has attribute $1$ and } n_{1}(G_i) >0 \\
 0, & \text{if the units with exposure $1$ in group design $G_i$}\\
 &\text{has attribute $0$ and } n_{1}(G_i) >0 \\
 \text{anything} & \text{otherwise}.
 \end{cases}
 \end{equation*}
 And similarly,
 \begin{equation*}
 A_2(G_i) =
 \begin{cases}
 1, & \text{if the units with exposure $2$ in group design $G_i$}\\
 &\text{has attribute $1$ and } n_{2}(G_i) >0 \\
 0, & \text{if the units with exposure $2$ in group design $G_i$}\\
 &\text{has attribute $0$ and } n_{2}(G_i) >0 \\
 \text{anything} & \text{otherwise}.
 \end{cases}
 \end{equation*}

 Therefore the heuristic for maximizing power of the Fisherian inference can be translated as the following optimization problem. 
 \begin{align*}
 \label{optimal}
 \operatorname{argmax}_{\{n_i\}} & \sum_{G_i \in \mathcal{G^*}} n_i (m_{1}(G_i) + m_{2}(G_i))\\
 \operatorname{s.t.} & \sum_{G_i \in \mathcal{G^*}} n_i c_1(G_i) \leq n_{[1]} \\
 & \sum_{G_i \in \mathcal{G^*}} n_i c_0(G_i) \leq N-n_{[1]}\\
 & \sum_{G_i \in \mathcal{G^*}} n_i m_{1}(G_i) A_1(G_i) \leq \eta \cdot \sum_{G_i \in \mathcal{G^*}} n_i m_{2}(G_i) A_2(G_i)\\
 & \sum_{G_i \in \mathcal{G^*}} n_i m_{2}(G_i) A_1(G_i) \leq \eta \cdot \sum_{G_i \in \mathcal{G^*}} n_i m_{1}(G_i) A_2(G_i)\\
 & \sum_{G_i \in \mathcal{G^*}} n_i m_{1}(G_i) (1-A_1(G_i)) \leq \eta \cdot \sum_{G_i \in \mathcal{G^*}} n_i m_{2}(G_i) (1-A_2(G_i))\\
 & \sum_{G_i \in \mathcal{G^*}} n_i m_{2}(G_i) (1-A_1(G_i)) \leq \eta \cdot \sum_{G_i \in \mathcal{G^*}} n_i m_{1}(G_i) (1-A_2(G_i))\\
 & n_i  \geq 0, n_i \in \mathbb{Z}, \forall n_i,
 \end{align*}
 where $\eta  = 1 + \epsilon$ for some $\epsilon >0$ that can be chosen to achieve a satisfiable trade-off between the two objectives of maximizing total number and balancing. 
 This is in the standard form of an integer linear programming problem or a knapsack problem in particular. The general case for attribute value set $|\mathcal{A}| > 2$ can be extended directly from this binary attribute case. 
\begin{remark}
 The heuristics for maximizing power is qualitative and hence the above optimization problem is just one of many ways to realize the heuristic. For example, the tuning parameter $\eta$ can be adjusted by the practitioner to achieve different tradeoffs for maximizing number of units with target treatment and balancing between the two treatments. Different values of $\eta$ can also be used for different balancing constraints as well. 
\end{remark}
\begin{remark}
 Integer linear programing problems are NP-hard and there are established iterative solvers that yield good approximations of the true optimizer. However, in this case, we can get fairly good approximation of the optimal assignment by simply taking one step of linear programming relaxation and rounding downwards. That is, we drop the constraint that $n_i \in \mathbb{Z}$ and solve the simple linear programing problem. Since we are rounding downwards and all coefficients are non-negative, the round-off integer solution is still feasible. This one-step linear relaxation has the advantage that it gives a fast initialization yielding near optimal power among all possible initializations. In particular, it does not scale with the number of units or group sizes as other methods do. 
\end{remark}

\section{Simulation set up}
We compare the power for different initializations leading to different designs. Given a fixed attribute vector $A$, different initializations of latent assignments $C$ will result in different compositions of exposures that are later permuted in the randomization test in Proposition \ref{validity}. Specifically, we want to compare the optimal design described in Section 5 derived from linear programming with random initializations and rejection sampling. A random initialization takes some fixed group assignment and external intervention assignment and permutes them randomly and separately. We also consider two rejection sampling methods for number of iterations $M = 10$ and $1000$. A rejection sampling method in our setting can be described in the following steps. 
\begin{enumerate}
\item generate a random initialization of latent assignments $(L,W)$, and compute the number of units with two target exposures under different attribute classes. Denote $n_{ij}$ the number of units with attribute $i$ and exposure equals target exposure $j$.
\item Repeat for $M$ iterations:
\begin{enumerate}
\item permute $L' = \pi_1 \cdot L$ and $W' = \pi_2 \cdot W$, for $\pi_1, \pi_2 \in \operatorname{Unif}(S)$
\item compute the number of units with target exposures under permuted latent assignments, and denote by $n'_{ij}$. Accept and assign $(L,W) \leftarrow (L',W')$ if 
\begin{align*}
& n_{01} + n_{11}+ n_{02} + n_{12}  < n'_{01} + n'_{11}+ n'_{02} + n'_{12} \\
& 1/\eta \leq n'_{01}/n'_{02} \leq \eta\\ 
& 1/\eta \leq n'_{11}/n'_{12} \leq \eta.
\end{align*}
\end{enumerate}
\end{enumerate}
The result of our simulations is shown in Figure~\ref{fig:simul-figures}. It can be seen that optimal design using linear programming yields higher power than the other initializations for all additive effects and group sizes. The advantage of linear programming is significantly more pronounced when group size increases slightly. This is because as group size increases, the number of possible exposures increases significantly and it is much more difficult for brute force algorithms with a fixed number of iterations to find a near-optimal solution.

\end{document}